\documentclass[sn-aps]{sn-jnl}

\usepackage{graphicx}%
\usepackage{multirow}%
\usepackage{amsmath,amssymb,amsfonts}%
\usepackage{amsthm}%
\usepackage{mathrsfs}%
\usepackage[title]{appendix}%
\usepackage{xcolor}%
\usepackage{textcomp}%
\usepackage{manyfoot}%
\usepackage{booktabs}%
\usepackage{algorithm}%
\usepackage{algorithmicx}%
\usepackage{algpseudocode}%
\usepackage{listings}%
\usepackage{upquote}

\usepackage[T1]{fontenc}

\usepackage{diagbox}

\usepackage{arydshln}

\usepackage{hyperref}

\hypersetup{
  hidelinks,
  colorlinks=true,
  allcolors=black,
  pdfstartview=Fit,
  breaklinks=true
}

\usepackage{xspace}

\newcommand{\st}{\ensuremath{\ | \ }}
\newcommand{\card}[1]{\left|#1\right|}
\renewcommand{\vec}[1]{\text{\mathversion{bold}$#1$}}

\newcommand{\Fq}{\mathbb{F}_q}
\newcommand{\Fqk}{\mathbb{F}_{q^k}}

\def\GL#1{\mathrm{GL}_{#1}}

\def\T{\mathbf T}

\def\H{\mathbf H}
\def\N{\mathbf N}

\newcommand{\rsp}{\mathrm{rowsp}}
\newcommand{\stab}{\mathrm{Stab}}
\newcommand{\orb}{\mathrm{Orb}}

\newcommand{\mx}{\mathrm{x}}

\newcommand{\cU}{\mathcal{U}}
\newcommand{\cV}{\mathcal{V}}
\newcommand{\cC}{\mathcal{C}}
\newcommand{\cG}{\mathcal{G}}

\usepackage[colorinlistoftodos,textwidth=20mm]{todonotes}
\theoremstyle{thmstyleone}%
\newtheorem{theorem}{Theorem}%
\newtheorem{lemma}{Lemma}%

\theoremstyle{thmstyletwo}%
\newtheorem{remark}{Remark}%

\theoremstyle{thmstylethree}%

\renewcommand{\theenumi}{(\alph{enumi})}

\begin{document}


\title[Spread Codes from Abelian
non-cyclic groups]{Spread Codes from Abelian
non-cyclic groups}

\author[1,2]{\fnm{Joan-Josep} \sur{Climent}}\email{jcliment@ua.es}
\equalcont{These authors contributed equally to this work.}

\author*[1,3]{\fnm{Ver\'{o}nica} \sur{Requena}}\email{vrequena@ua.es}

\author[1,4]{\fnm{Xaro} \sur{Soler-Escriv\`{a}}}\email{xaro.soler@ua.es}
\equalcont{These authors contributed equally to this work.}

\affil*[1]{\orgdiv{Departament de Matem\`{a}tiques}, \orgname{Universitat d'Alacant}, \orgaddress{\street{Carretera de Sant Vicent del Raspeig s/n}, \city{Sant Vicent del Raspeig}, \postcode{E-03690}, \country{Spain}}}
\affil[2]{orcidID: 0000-0003-0522-0304}
\affil[3]{orcidID: 0000-0002-1497-6456}
\affil[4]{orcidID: 0000-0001-7595-7032}


\abstract{Given the finite field $\mathbb{F}_{q}$, for a prime power $q$, in this paper we present a way of constructing spreads of $\mathbb{F}_{q}^{n}$.
They will arise as orbits under the action of an Abelian non-cyclic group. 
First, we construct a family of orbit codes of maximum distance using this group, and then we complete each of these codes to achieve a spread of the whole space
having an orbital structure.}


\keywords{Network coding; subspace codes; Grassmannian; spreads; group action; general linear group}


\pacs[MSC Classification]{11T71, 94B60, 20H30, 05E20, 68P30}


\maketitle


\section{Introduction}

\emph{Network coding} is a part of information theory that describes a method to maximize the rate of a network which is modelled by a directed acyclic multigraph, with one or multiple sources and multiple receivers. 
First introduced in \cite{Ahlswede2000}, the key point of this method is allowing the intermediate nodes of the network to transmit linear combinations of the inputs they receive. 
The algebraic approach given by K\"otter and Kschischang in \cite{Koetter2008} provided a rigorous mathematical setup for error correction when coding in non-coherent networks and, as a result, this theory was able to advance vastly. 
In this setting, the transmitted messages (\emph{codewords}) are vector subspaces of a given vector space $\mathbb{F}^n_q$, where $\mathbb{F}_q$ is the finite field of $q$ elements and a \emph{subspace code} is just a collection $\mathcal{C}$ of vector subspaces of $\mathbb{F}_{q}^{n}$. 
When all subspaces of $\mathcal{C}$ have the same dimension, we say that $\mathcal{C}$ is a \emph{constant dimension code}.
The minimum distance $d(\mathcal{C})$ of $\mathcal{C}$ is computed in the usual way by using a metric called \emph{subspace distance} in the set of all subspaces of $\mathbb{F}_q^n$.
We refer the reader to \cite{Bassoli2013, Trautmann2018} and references therein for further information regarding network coding and subspace codes.

One of the most important and studied families of subspace codes are \emph{spread codes} (or simply \emph{spreads}). 
A spread code is a constant dimension code such that all its elements intersect pairwise trivially and their union covers the whole vector space. 
Spreads are clearly a relevant family of constant dimension codes 
since they reach the maximum distance and, at the same time, the maximum size for that distance. 
For this reason, many papers in the literature about subspace codes are devoted to the study and construction of this type of codes (see \cite{Gluesing-Luerssen2019, Gorla2012, Honold2018, Manganiello2008} for instance). 

A relevant way of constructing constant dimension codes is by considering the natural action of the general linear group, $\GL{n}(\Fq)$, on the Grassmannian $\cG_q(k, n)$, which is the set of all $k$-dimensional vector subspaces of $\mathbb{F}_{q}^{n}$ (for an integer $k\in \{1,\dots, n\}$). 
Using this technique, the codes arise as orbits under the action of some specific subgroup of $\GL{n}(\Fq)$. 
Constant dimension codes constructed in this way are called \emph{orbit codes}. 
In the linear network coding setting, they were first introduced in \cite{Trautmann2010}, where  their main properties are given. 
Due to the group action point of view, orbit codes have a nice mathematical structure, and they have been investigated by many different authors since then (see for instance \cite{Bardestani2015,Gluesing-Luerssen2015,Manganiello2011,Rosenthal2013,Trautmann2014}).
Using powerful tools from group theory, the distance of this kind of codes can be calculated in a simpler way, and we can compute the size of the code in terms of the order of the acting group and the order of the corresponding stabilizer subgroup \cite{Trautmann2010}. 
In addition, there exist different algorithms for decoding orbit codes \cite{Poroch2019, Trautmann2014} and several of the known algebraic constructions of constant dimension codes can be seen as orbit codes.
Most of the research on this topic focus on the use of cyclic subgroups of $\GL{n}(\Fq)$, in which case we speak about \emph{cyclic orbit codes}. 
In particular, in \cite{Trautmann2010} appears the first construction of spreads with an orbital structure. 
From here, we can find several works on spreads with an orbital structure provided by a cyclic group (for instance, \cite{Chen2020b,Trautmann2013b}). 

While research on cyclic orbit codes abounds, the same is not true when we want to focus on other types of subgroups of the general linear group. 
As a first step, in \cite{Climent2019b} we approach the study of orbit codes through the action of Abelian non-cyclic subgroups of $\GL{n}(\Fq)$, giving a specific construction of maximum distance.
Pursuing this line of research, the papers \cite{Chen2021b,TerraBastos2020} are also concerned with Abelian non-cyclic orbit codes. 
Nevertheless, as far as we know, the only construction on spreads through the action of a non-cyclic Abelian group is given in \cite{Chen2021b}. 
In this paper, the authors construct an Abelian non-cyclic orbit code of $\mathbb{F}^{2k}_q$ of dimension $k$ having maximum distance and then obtain a $k$-spread of $\mathbb{F}^{2k}_q$ just by adding two $k$-subspaces of $\mathbb{F}^{2k}_q$. 

Our main objective in this paper is to pursue the research of orbit codes constructed by using non-cyclic Abelian groups. 
Specifically, we generalize the results obtained in \cite{Chen2021b} in the following sense: For an even integer $n$ and $k$ a divisor of $n$, we firstly construct an Abelian non-cyclic orbit code of $\mathbb{F}^{n}_q$ of dimension $k$ having maximum distance. 
Then, we achieve to complete this orbit code with a nice family of $k$-subspaces of $\mathbb{F}^{n}_q$ in such a way the resulting code is a $k$-spread of $\mathbb{F}^{n}_q$ with an orbital structure which is not cyclic. 
This generalization is not immediate and has required the use of new techniques, not used in \cite{Chen2021b}.

The paper is structured as follows. 
In Section~\ref{sec:prelim}, we collect all the background on finite fields and subspace codes that will be needed in the subsequent sections. 
Section~\ref{sec:const} is devoted to our orbital constructions of maximum distance codes and is divided into two parts. 
Firstly, we construct a non-cyclic Abelian group $\mathbf{H}$ and from it a maximum distance family of orbit codes of dimension $k$ in $\mathbb{F}^{n}_q$. 
Secondly, we tackle the construction of $k$-spreads in $\mathbb{F}^{n}_q$ by completing each orbit code previously obtained. To do so, we define two different families of codes, one of them of orbital type. The final $k$-spread is then presented as a disjoint union of these three types of codes.  
Finally, Section~\ref{sec:concl} concludes the paper.

\section{Preliminaries}
\label{sec:prelim}

\subsection{Finite fields}\label{subsec:cossos}

In this section, we recall some basic notions and results concerning finite fields that will be needed later.
The reader can find more details in any basic book on this subject, e.g.~\cite{Lidl1986bk}.

Given a finite field $\Fq$ and a positive integer $n$, we will denote by $\Fq^{n \times n}$ the set of all $n \times n$ matrices with entries in $\Fq$ and by $\GL{n}(\Fq)$ the \emph{general linear group} of degree $n$ over $\Fq$. 

Assume that $n=ks$, for some positive integers $k, s$ and consider a \emph{primitive element} $\alpha$ of the field $\Fqk$, that is, a generator of the cyclic group $\Fqk^*$.
The \emph{companion matrix} of the minimal polynomial $p(\mx)=a_0+a_1\mx+\cdots+a_{k-1}\mx^{k-1}+\mx^k$ of $\alpha$ over $\Fq$ is the matrix
\[
  M_k
  = 
  \begin{pmatrix}
    0      & 1       &  0      & \cdots &  0       \\
    0      & 0       &  1      & \cdots &  0       \\
    \vdots &  \vdots &  \vdots & \ddots &  \vdots  \\
    0      & 0       &  0      & \cdots &  1       \\
   -a_0    & -a_1    & -a_2    & \cdots & -a_{k-1}
  \end{pmatrix}
  \in 
  \GL{k}(\Fq).
\]
It turns out that $p(\mx)$ is the characteristic polynomial of $M_k$ and, then, $\Fqk=\Fq[\alpha]$ can be realized as a set of matrices of $\Fq^{k\times k}$ through the following field isomorphism
\begin{equation}\label{eq:def_phi}
\begin{array}{cccc}
  \phi : & \Fq[\alpha]                 &  \longrightarrow & \Fq\left[ M_k \right] \\
         & \sum_{i=0}^{k-1}b_i\alpha^i & \mapsto      & \sum_{i=0}^{k-1}b_iM_k^i. 
\end{array}
\end{equation}
Therefore, $M_k$ can be seen as a primitive element of the finite field $ \Fq\left[ M_k \right]$. 
Equivalently, the multiplicative order of $M_k$, denoted by $o(M_k)$, is $q^k-1$ and
\[
  \Fq\left[ M_k \right]
  =
  \{0_{k\times k}\}\cup\langle M_k\rangle,
\]
where $0_{k\times k}$ is the null matrix of $\Fq^{k\times k}$ and
\[
  \langle M_k\rangle
  =
  \big\{I_k, M_k, M_k^2, \dots, M_k^{q^k-2}\big\}
\]
is the multiplicative group generated by $M_k$.

On the one hand, the isomorphism $\phi$ provided in expression \eqref{eq:def_phi} allows to map vector subspaces of $\Fqk^s$ into vector subspaces of $\Fq^n$ (this is the well-known \emph{field reduction} technique, see \cite{Lavrauw2015} for instance). 
Specifically, each line generated by a vector $(u_1,\dots, u_s)$ of $\Fqk^s$, will produce a subspace of dimension $k$ of $\Fq^n$ with the following injective map
\begin{equation}\label{eq:def_varphi}
\begin{array}{cccc}
\varphi: & \mathcal{G}_{q^k}(1,s)                &  \longrightarrow & \mathcal{G}_{q}(k,n) \\
         &\rsp \begin{pmatrix}
                 u_{1}    &  \dots & u_{s}
               \end{pmatrix}       
                                              & \mapsto          & \rsp \begin{pmatrix}
                                                                          \begin{array}{c|c|c}
                                                                            \phi(u_{1}) &   \dots & \phi(u_{s})
                                                                          \end{array}
                                                                        \end{pmatrix}.
\end{array}
\end{equation}

On the other hand, $\phi$ is also useful to embed $\GL{s}(\Fqk)$ into $\GL{n}(\Fq)$. 
By using it, we obtain the following group monomorphism
\begin{equation}\label{eq:def_psi}
\begin{array}{cccc}
  \psi : & \GL{s}(\Fqk)                &  \longrightarrow & \GL{n}(\Fq) \\
         & \begin{pmatrix}
             a_{11} & \cdots & a_{1s} \\
             \vdots & \ddots & \vdots \\
             a_{s1} & \cdots & a_{ss}
           \end{pmatrix}       
                                       & \mapsto          & \begin{pmatrix}
                                                              \begin{array}{c|c|c}
                                                                \phi(a_{11}) & \cdots & \phi(a_{1s}) \\ \hline
                                                                \vdots       & \ddots & \vdots       \\ \hline
                                                                \phi(a_{s1}) & \cdots & \phi(a_{ss})
                                                              \end{array}
                                                            \end{pmatrix}.
\end{array}
\end{equation}

\subsection{Subspace codes}
\label{sec:subspace codes}

We present here the main results of subspace codes that we need to understand this work.

The Grassmannian $\cG_q(k, n)$ can be seen as a metric space endowed with the following \emph{subspace distance} (see \cite{Koetter2008})
\begin{equation*}
  d_S(\cU,\cV)
  =
  \dim(\cU+\cV)-\dim(\cU\cap \cV)
  =
  2(k -\dim(\cU\cap \cV)),
\end{equation*}
for all $\cU, \cV\in\cG_q(k, n)$. 

If $\cC$ is a non-empty subset of $\cG_q(k, n)$, then we say that $\cC$ is a \emph{constant dimension code} and its \emph{minimum distance} is defined as 
\begin{equation}\label{eq:distC}
  d_S(\cC)
  =
  \min
  \left\{ 
    d_S(\cU, \cV) 
    \ | \ 
    \cU, \cV\in \cC, \ \cU\neq \cV
  \right\}
  \leq 
  \begin{cases}
    2k,     & \text{if $2k\leq n$}, \\
    2(n-k), & \text{if $2k\geq n$}.
  \end{cases}
\end{equation}
If $|\cC|=1$, we put $d_S(\cC)=0$; in any other case, $d_S(\cC)>0$. When the upper bound in expression~\eqref{eq:distC} is attained, we say that $\cC$ is a \emph{constant dimension code of maximum distance}. 
Note that $d_S(\cC) = 2k$ only if $2k \leq n$ and all the codewords in $\cC$ intersect trivially.
This class of codes of maximum distance are known as \emph{partial spread codes}, and they were introduced in \cite{Gorla2014} as a generalization of the class of \emph{spread codes}, previously studied in \cite{Manganiello2008}. 
The code $\cC$ will be a spread code, if the subspaces in $\cC$ pairwise intersect trivially, and they cover the whole space $\Fq^n$ (see \cite{Hirschfeld1998bk}). 
These codes only occur in the case where $k \mid n$ and have cardinality $\frac{q^n-1}{q^k-1}$. 
In this way, it can be said that spread codes are partial spreads of maximum size, since the size of a partial spread code of dimension $k$ (or $k$-partial spread code) is always upper bounded by 
\begin{equation}\label{eq:cota_partial_spread}
  \frac{q^n-q^m}{q^k-1},
\end{equation}
where $m$ is the reminder obtained dividing $n$ by $k$ (see \cite{Gorla2014}). 
Notice that any code of lines $\cC\subseteq \mathcal{G}_q(1,n) $ with $|\cC|\geq 2$ is in particular a partial spread code with size $|\cC| \leq \frac{q^n-1}{q-1}$, whereas $\mathcal{G}_q(1,n)$ can be seen as the spread of lines of $\Fq^n$.
These two types of codes will be useful in the following sections.

In case that we have $\cC\subseteq \cG_q(k, n)$ with $2k\geq n$, then we can consider the \emph{dual code} of $\cC$, that is, the set $\cC^\perp=\{\mathcal{V}^\perp\ |\ \mathcal{V}\in \cC\}$, which is a constant dimension code of dimension $n-k$ with  the same cardinality and distance as $\cC$ (see \cite{Koetter2008}). 
In particular, if $d_S(\cC)=2(n-k)$, then $\cC^\perp$ is an $(n-k)$-partial spread code and the size of $\cC$ can be also upper bounded in terms of expression \eqref{eq:cota_partial_spread}. 
This is the reason why from now on we consider that $2k \leq n$.

An important class of constant dimension codes are those called 
\emph{orbit codes}, introduced in \cite{Trautmann2010}.
These codes are defined as orbits under the action of some subgroup of the general linear group. 
Consider $V\in \Fq^{k\times n}$ a full-rank matrix generating a subspace $\cV=\rsp(V)\in\cG_q(k, n)$. 
The map 
\[
\begin{array}{ccc}
  \mathcal{G}_q(k,n)  \times \GL{n}(\Fq) & \longrightarrow & \mathcal{G}_q(k,n) \\
  (\cV, A)                               & \mapsto         & \cV\cdot A=\rsp(VA),
\end{array}
\]
defines a group action from the right of the general linear group on $\cG_q(k, n)$ (see \cite{Trautmann2010}) since it is independent of the choice of the generating matrix $V$.
Given a subgroup $\H$  of $\GL{n}(\Fq)$, the \emph{orbit code} $\orb_{\H}(\cV)$ is the orbit generated by the action of $\H$ on $\cV$, that is,
\[
  \orb_{\H}(\cV)
  =
  \{ \cV\cdot A\ |\ A\in\H \}
  \subseteq
  \mathcal{G}_q(k,n).
\]
The size of this orbit code is $|\orb_{\H}(\cV)|=\frac{|\H|}{|\stab_{\H}(\cV)|}$, where $\stab_{\H}(\cV)=\{A\in\H\ | \ \cV\cdot A=\cV\}$ is the stabilizer subgroup of the subspace $\cV$ under the action of $\H$. 
If $\H=\stab_{\H}(\cV)$, then $\orb_{\H}(\cV)=\{\cV\}$ and $d_S(\orb_{\H}(\cV))=0$.
In any other case, the minimum distance of the orbit code $\orb_{\H}(\cV)$ can be calculated as (see \cite{Trautmann2010}) 
\[
  d_S(\orb_{\H}(\cV))
  =
  \min \{ d_S(\cV, \cV\cdot A) \ | \  A\in\H\setminus \stab_{\H}(\cV) \}.
\] 

Recall that we are assuming $n=ks$ and consider the field reduction map $\varphi$ defined in expression \eqref{eq:def_varphi} which maps lines of $\Fqk^s$ into vector subspaces of dimension $k$ of $\Fq^n$.
Since $\varphi$ is injective, it preserves intersections and, therefore, it follows that $d_{S}(\varphi(\cC))=kd_{S}(\cC)=2k$, for any $\cC \subseteq \cG_{q^k}(1,s)$. 
In other words, $\varphi(\cC)$ is a $k$-partial spread of $\Fq^{n}$, for any code of lines $\cC$ of $\Fqk^s$.
In particular, if we consider the spread of all lines of $\Fqk^{s}$, it turns out that 
\begin{equation}\label{def:spreadS}
  \varphi(\cG_{q^k}(1,s)) \subseteq \cG_{q}(k,n),
\end{equation}
is a $k$-spread of $\Fq^n$, which is called the \emph{Desarguesian} $k$-spread of $\Fq^n$ (see \cite{Lavrauw2015}). 
Originally due to Segre (see \cite{Segre1964}), in the network coding setting, this construction appears for the first time in \cite{Manganiello2008}.

Notice that the field reduction map $\varphi$, together with the group monomorphism $\psi$ defined in expression~\eqref{eq:def_psi}, allow us to establish a relation between the group action of $\GL{s}(\Fqk)$ on $\cG_{q^k}(1,s)$ and the group action of $\GL{n}(\Fq)$ on $\cG_q(k,n)$ as follows (see \cite{AlonsoGonzalez2021,Navarro-Perez2022})
\begin{equation} \label{eq:equiv_accions}
  \varphi(\cV\cdot A)
  =
  \varphi(\cV)\cdot \psi(A),
\end{equation}
for all $\cV\in\cG_{q^k}(1,s)$ and $A\in \GL{s}(\Fqk)$.
In particular, when we consider a subgroup $\H\leq \GL{s}(\Fqk)$ and an orbit code $\orb_{\H}(\cV)$, for some $\cV\in\cG_{q^k}(1,s)$, then 
\begin{align} \label{eq:S_orbSinger}
  \varphi(\orb_{\H}(\cV))
    & = \{\varphi(\cV\cdot A)\ |\ A\in\H \} \nonumber \\
    & = \{\varphi(\cV) \cdot \psi(A) \ |\ \psi(A)\in \psi(\H)\} \nonumber \\
    & = \orb_{\psi(\H)}(\varphi(\cV) ) \subseteq \cG_{q}(k,n).
\end{align}

\section{Our construction}
\label{sec:const}

In the rest of the paper, we consider $n,k,s,t$ positive integers such that $n=ks$, with $s=2t$ and $\gcd(t,q^k-1)=1$.

The main purpose of this section is to show how we can construct the Desarguesian $k$-spread code of $\mathbb{F}^{n}_q$ given in expression~\eqref{def:spreadS} starting from the action of a specific Abelian non-cyclic subgroup $\bar{\H}$ of $\GL{n}(\Fq)$. 
To do so, in Subsection \ref{sec:3.1} we will work first with an appropriate subgroup $\H$ of $\GL{s}(\Fqk)$ and its action on lines of $\Fqk^s$ and then we will use the field reduction technique explained in Subsection \ref{sec:subspace codes} to build our subgroup $\bar{\H}$ of $\GL{n}(\Fq)$. 
Then we will construct a $k$-partial spread of $\mathbb{F}^{n}_q$ as an orbit generated by the action of $\bar{\H}$. 
Following this, in Subsection \ref{sec:3.2} we will explain how to achieve the whole $k$-spread.

\subsection{A $k$-partial spread with an orbital structure}
\label{sec:3.1}

In this subsection we provide an explicit description of a $k$-partial spread of $\mathbb{F}^{n}_q$ with an orbital structure. 
For this purpose, firstly, we will construct an Abelian non-cyclic subgroup $\H$ of $ \GL{s}(\Fqk)$ with a suitable action on certain lines of $\Fqk^s$. 

Let $M_t\in \GL{t}(\Fqk)$ be the companion matrix of a primitive polynomial of degree $t$ over $\mathbb{F}_{q^{k}}$. 
As we explain in Subsection~\ref{subsec:cossos}, we have that $\mathbb{F}_{q^{kt}}\cong \mathbb{F}_{q^k}[M_t]$ and, therefore, the order of $M_t$ is $o(M_t)=q^{kt}-1$. 
Consider $C:=M_t^{q^k-1}$ whose multiplicative order, clearly, is $r:=\frac{q^{kt}-1}{q^k-1}$; and, let $\alpha\in \Fqk$ be a primitive element, whose multiplicative order is $q^{k}-1$.

We construct the matrices $h_{1}, h_{2} \in \GL{s}(\Fqk)$ given by
\begin{equation}\label{eq:matriushi}
h_1:= \begin{pmatrix}
        C    &  I_t\\
        0_{t\times t}    &   \alpha I_t
    \end{pmatrix}
    \quad \text{and} \quad
h_2:= \begin{pmatrix}
         \alpha I_{t}    &  -I_{t}\\
         0_{t\times t}    &  C
      \end{pmatrix}.
\end{equation}

The following result will be useful in order to compute the multiplicative order of $h_1$ and $h_2$. 

\begin{lemma} \label{lemma:1}
For any positive integer $\ell$, one has that
\[
  \gcd(\ell,q-1)=1 
  \ \text{if and only if} \ 
  \gcd \left(\frac{q^{\ell}-1}{q-1},q-1\right)=1.
\]
\end{lemma}
\begin{proof}
First we notice that if $p$ is a prime dividing $q-1$, then $q \equiv 1 \pmod{p}$ and so $q^{i} \equiv 1 \pmod{p}$ for $i \geq 1$. 
Therefore,
\[
  \frac{q^\ell-1}{q-1}
  =
  q^{\ell-1}+q^{\ell-2}+ \cdots + q+1
  \equiv
  \ell \pmod{p}.
\]
Then, it is clear that 
$p \mid \ell$ if, and only if, $p \mid \frac{q^{\ell}-1}{q-1}$
and the result follows. 
\end{proof}

Note that, since we are assuming that $\gcd(t,q^k-1)=1$, Lemma~\ref{lemma:1} states that $\gcd(r,q^k-1)=1$.

The following result is a generalization of Lemmas~3.1, 3.2 and 3.3 of \cite{Chen2021b}. 
The proof runs analogously, and thus we omitted it. 

\bigskip

\begin{lemma} 
\label{th1}
Consider the matrices $h_1, h_2$ defined in expression \eqref{eq:matriushi}.
The following statements are satisfied:
\begin{enumerate}
\item The multiplicative order of $h_1$ and $h_2$ is $q^{kt}-1$.
\item The matrices $h_1$ and $h_2$ commute.
\item $\langle h_1\rangle \cap \langle h_2\rangle=\{I_{s}\}$.
\end{enumerate}
\end{lemma}

From the previous lemma, it follows that the group  generated by the matrices $h_1$ and $h_2$ is an Abelian non-cyclic subgroup of $\GL{s}(\Fqk)$. 
Let us denote it us $\H:=\langle h_1, h_2\rangle=\langle h_1\rangle \langle h_2\rangle$. 
Thus, the order of $\H$ is $(q^{kt}-1)^2$ and we can express its elements as,
\[
  \H
  =
  \langle h_{1} \rangle \langle h_{2} \rangle
  =
  \left\{h_1^a h_2^b \st 1\leq a, b \leq q^{kt}-1 \right\}.
\]
Moreover, any arbitrary element of $\H$ has the following matrix expression
\begin{equation} \label{eq:ab}
  h_1^a h_2^b
  =
  \begin{pmatrix}
    \alpha^bC^a   & D_{a,b} \\
    0_{t\times t} & \alpha^aC^b
  \end{pmatrix}
\end{equation}
where 
\begin{equation}\label{eq:D_ab}
  D_{a,b}
  :=
  \sum_{j=1}^a \alpha^{j-1}C^{a+b-j}
  -
  \sum_{j=1}^b \alpha^{j-1}C^{a+b-j}
\end{equation}
and 
$a, b\in\{1,\dots, q^{kt}-1\}$. 
We will denote by $(h_1^a h_2^b)_i$ the $i$-th row of the matrix $h_1^a h_2^b\in \H$, for any $i\in\{1,\dots, s\}$.

Let us denote $\vec{e}_{i}\in\Fqk^s$ the $i$-th canonical vector, for any $i\in\{1,\dots, s\}$. 
We are interested in the action of the group $\H$ on the lines of $\Fqk$ generated by $\vec{e}_i$, for $i\in\{1,\dots, t\}$. 
For this reason, we analyse the corresponding stabilizer subgroup. 
First, we need the following technical lemma.

\begin{lemma} \label{lem:ab}
Consider the matrix $D_{a,b}\in \Fqk[M_t]$ defined in expression~\eqref{eq:D_ab}, for $a,b\in\{1,\dots ,q^{kt}-1$\}, then 
\[
  D_{a,b}
  =
  0_{t\times t} 
  \ \text{if and only if} \
  a=b.
\]
\end{lemma}
\begin{proof}
Clearly, if $a=b$ then the matrix $D_{a,b}$ is the null matrix. 
Conversely, assume that $D_{a,b} = 0_{t \times t}$ and let us see that necessarily $a=b$.

Arguing by contradiction, if $a>b$, then we can simplify the matrix $D_{a,b}$ in the following way
\begin{align*}
  0_{t\times t}
  & = D_{a,b} 
    = \sum_{j=b+1}^a \alpha^{j-1}C^{a+b-j}  \\
  & = \alpha^bC^{a-1}+\alpha^{b+1}C^{a-2} + \cdots + \alpha^{a-1}C^{b} \\
  & = \alpha^bC^b \left( C^{a-b-1}+\alpha C^{a-b-2}+ \cdots+\alpha^{a-b-1}I_t \right).
\end{align*}
Notice that  $\alpha^b C^b\in\Fqk[M_t]$ is a regular matrix, since it is non-zero and $\Fqk[M_t]$ is a field. 
Thus, we obtain that 
\begin{align*}
  0_{t\times t} 
  & = C^{a-b-1}+\alpha C^{a-b-2}+ \cdots+\alpha^{a-b-1}I_t \\
  & = \alpha^{a-b-1}( (\alpha^{-1}C)^{a-b-1} + (\alpha^{-1}C)^{a-b-2}+\cdots + (\alpha^{-1}C) + I_t).
\end{align*}
As a consequence, the right factor of this last product must be zero. 
But this is a geometric series of ratio $\alpha^{-1}C$. 
Let us see that $\alpha^{-1}C\neq I_t$. 
If $\alpha^{-1}C=I_t$, then $C = \alpha I_t$ and $r=o(C)=o(\alpha)=q^k-1$.
But this is a contradiction with the fact that $\gcd(q^k-1,r)=1$. 
Thus, we can obtain the sum of this geometric series and write
\begin{align*}
  0_{t\times t}
  & = (\alpha^{-1}C)^{a-b-1}+ (\alpha^{-1}C)^{a-b-2}+\cdots + (\alpha^{-1}C) + I_t \\
  & = \left( (\alpha^{-1}C)^{a-b}-I\right)\left( \alpha^{-1}C-I\right)^{-1}.
\end{align*}
Therefore, we obtain that $C^{a-b}=\alpha^{a-b}I_t$.
It means that
\[
  \frac{o(C)}{\gcd\left(o(C),a-b\right)}
  =
  \frac{o(\alpha)}{\gcd\left(o(\alpha),a-b\right)},
\]
that is, 
\[
  \frac{r}{\gcd\left(r,a-b\right)}
  =
  \frac{q^k-1}{\gcd\left(q^k-1,a-b\right)}
  =
  1
\]
since $\gcd(q^k-1,r)=1$. 
Thus, both $r$ and $q^k-1$ must divide $a-b$. 
And then $q^{kt}-1=r(q^k-1)$ divides $a-b>0$. 
But this is not possible, since $a,b\in\{1,\dots, q^{kt}-1\}$.
 
If $a<b$ then we have an analogous case to the previous one. 
\end{proof}

As a consequence of this result, we are able to obtain the stabilizer subgroup of $\H$ corresponding to the lines generated by $\vec{e}_i$, for $i\in\{1,\dots, t\}$.

\bigskip

\begin{theorem}\label{lem:stabH}
For all $i\in\{1,\dots, t\}$, one has that
\[
  \stab_{\H}(\rsp(\vec{e}_{i}))
  =
  \langle (h_1h_2)^r\rangle 
  = 
  \langle \alpha I_s\rangle.
\]
\end{theorem}
\begin{proof}
Let $h_1^ah_2^b \in \stab_{\H}(\rsp(\vec{e}_{i}))$, for some $a,b\in\{1,\dots, q^{kt}-1\}$. 
Then we have that $\rsp(\vec{e}_{i})=\rsp((h_1^ah_2^b)_i)$. 
And this happens exactly when the $i$-th row of the matrix $h_1^ah_2^b$ has the form $\alpha^m \vec{e}_i$, for some $m\in\{1,\dots,q^k-1\}$.
Therefore, since $i\in \{1,\dots , t\}$, considering the general expression of any element of $\H$ given in expression \eqref{eq:ab},
it follows that the $i$-th row of the submatrix $D_{a,b}$ of $h_1^ah_2^b$ is zero. 
Since $D_{a,b}$ is a matrix in the field $\Fqk[M_t]$, it follows that it must be the null matrix. 
As a consequence, from Lemma~\ref{lem:ab}, we can write 
\[
  h_1^ah_2^b
  =
  h_1^ah_2^a
  =
  \begin{pmatrix}
    \alpha^a C^a   &  0_{t\times t} \\
    0_{t\times t}  & \alpha^a C^a
  \end{pmatrix},
  \quad
  \text{for some $a\in\{1,\dots , q^{kt}-1\}$}.
\]
Moreover, the $i$-th row of $h_1^ah_2^a$ has the form $\alpha^m \vec{e}_{i}$ for some $m\in\{1,\dots, q^k-1\}$. 
In particular, the $i$-th row of $\alpha^a C^a-\alpha^mI_t$ is zero and then this is a non-regular matrix of the field $\Fqk[M_t]$. 
Thus, $\alpha^a C^a-\alpha^mI_t=0_{t\times t}$, that is, $\alpha^a C^a=\alpha^mI_t=(\alpha I_t)^m$. 
Then we have
\[
  h_1^ah_2^a
  =
  \begin{pmatrix}
    \alpha^a C^a   & 0_{t\times t} \\
    0_{t\times t}  & \alpha^a C^a
  \end{pmatrix}
  =
  \begin{pmatrix}
    (\alpha I_t)^m   & 0_{t\times t}\\
    0_{t\times t}    & (\alpha I_t)^m
  \end{pmatrix}
  \in 
  \langle \alpha I_{s} \rangle.
\]
Therefore, we obtain the equality
\[
  \stab_{\H}(\rsp(\vec{e}_{i}))
  = 
  \langle \alpha I_s\rangle,
\]
since the other inclusion trivially holds. 
Finally, since $o(C)=r$, one has that $(h_1h_2)^r=\alpha^rI_s$ and then $\langle (h_1h_2)^r\rangle \subseteq \langle \alpha I_s\rangle$. 
Thus, we obtain the equality, since both groups have order $q^k-1$.
\end{proof}

\begin{remark}\label{rem:singer}
A well-known result of Singer (1938) about the action of the subgroup $\langle M_t\rangle\subseteq \GL{t}(\Fqk)$ on the Grassmannian of lines $\cG_{q^k}(1,t)$ states, among other things, that 
\[
  \stab_{\langle M_t\rangle}(\ell)=\langle \alpha I_t\rangle,
\]
for any line $\ell \in \cG_{q^k}(1,t)$ (see \cite[Th. 6.2]{Beth1999bk}). 
This fact will be useful in Subsection \ref{sec:3.2}. 
In Theorem \ref{lem:stabH} we have obtained a similar result for our group $\H$ when it acts on some lines of $(\Fqk)^s$. 
We would like to point out that our group $\H$ is far from verifying such a result for any line of $(\Fqk)^s$. 
Thus, for example, let us note that the subgroup $\langle h_1\rangle$ of $\H$ will be in the stabilizer of lines generated by vectors $\vec{e}_{i}$, when $i\in\{t+1,\dots, s\}$.
\end{remark}

\bigskip

Now we are ready to consider the following $1$-dimensional orbit codes of $\Fqk^s$, for  $i\in\{1,\dots, t\}$
\begin{equation}\label{eq:Ci}
  \mathcal{C}_i
  :=
  \orb_{\H}(\rsp(\vec{e}_{i}))
  =
  \{\rsp((h_1^a h_2^b)_i)\ |\ 1\leq a, b\leq q^{kt}-1 \}\subseteq \cG_{q^k}(1,s).
\end{equation}
From Theorem \ref{lem:stabH}, the size of each orbit code $\cC_i$ will be
\[
  |\mathcal{C}_i|=\frac{|\H|}{|\stab_{\H}(\rsp(\vec{e}_{i}))|}=\frac{(q^{kt}-1)^2}{q^k-1}.
\]

As explained in Subsection \ref{sec:subspace codes}, from the lines $\rsp(\vec{e}_{i})\in \Fqk^s$, for $i\in\{1,\dots, t\}$, and the group $\H\subseteq \GL{s}(\Fqk)$, the injective map $\varphi$ defined in expression~\eqref{eq:def_varphi} and the group monomorphism $\psi$ defined in expression~\eqref{eq:def_psi}, allow us to use the codes $\mathcal{C}_i$ given in expression~\eqref{eq:Ci} in order to construct constant dimension codes of $\cG_{q}(k,n)$. 
Starting from the field isomorphism $\phi$ defined in expression~\eqref{eq:def_phi}, one has that $\phi(0)=0_{k\times k}$ and $\phi(1)=I_{k}$. 
Thus, we consider the vector subspace of $\Fq^n$
\[
  \mathcal{U}_{k,i}:=\varphi(\rsp(\vec{e}_{i}))
  =
  \rsp 
  \begin{pmatrix}
    \begin{array}{c|c|c|c|c}
      0_{k\times k} & \dots & I_{k}& \dots & 0_{k\times k}
    \end{array}
  \end{pmatrix}
  \subseteq 
  \cG_{q}(k,n).
\]
Moreover, we also consider the group $\bar{\H}:=\psi(\H)\subseteq \GL{n}(\Fq)$.
Now, according to expressions \eqref{eq:equiv_accions} and \eqref{eq:S_orbSinger}, for any $i\in\{1,\dots, t\}$, we have that
\begin{align}
  \bar{\mathcal{C}}_i
  &
  :=
  \varphi(\mathcal{C}_i)
  =
  \varphi(\orb_{\H}(\rsp(\vec{e}_{i}))) \nonumber \\
  &
  =
  \orb_{\psi(\H)}(\varphi(\rsp(\vec{e}_{i})))
  =
  \orb_{\bar{\H}}(\mathcal{U}_{k,i}).
  \label{eq:Cibarra}
\end{align}

Since $\varphi$ is an injective map and $d_S (\mathcal{C}_i)=2$, one has that $\bar{\mathcal{C}}_i$ is a $k$-partial spread of  $\Fq^n$, that is, it has dimension $k$ and minimum distance $d_S (\bar{\mathcal{C}}_i)=kd_S (\mathcal{C}_i)=2k$. Moreover, 
\[
|\bar{\mathcal{C}}_i|=|\mathcal{C}_i|=\frac{(q^{kt}-1)^2}{q^k-1}=(q^{kt}-1)r.
\]
Since a $k$-spread of $\Fq^n$ has size $\frac{q^{n}-1}{q^k-1}$ (see Section~2), we can calculate how far each of the $\bar{\mathcal{C}}_i$ codes is from being a $k$-spread, for any $i\in\{1,\dots , t\}$. Specifically, we obtain that 
\[
|\bar{\mathcal{C}}_i| + 2r = (q^{kt}-1)r + 2r = r(q^{kt}+1)= \frac{q^{kt}-1}{q^{k}-1}(q^{kt}+1)=\frac{q^{n}-1}{q^{k}-1}. 
\]

\subsection{Achieving a $k$-spread from each $\bar{\mathcal{C}}_i$}
\label{sec:3.2}

In this section, we explain how we can obtain the Desarguesian $k$-spread of $\Fq^n$ given in expression~\eqref{def:spreadS} starting from each $k$-partial spread $\bar{\mathcal{C}}_i$ defined in expression~\eqref{eq:Cibarra}, for any $i\in \{1,\dots,t\}$. 
That is, fixing a $k$-partial spread $\bar{\mathcal{C}}_i$, we explicitly construct $2r$ subspaces of $\Fq^n$ having dimension $k$ and trivial intersection between them and also with the subspaces of $\bar{\mathcal{C}}_i$. 

To do so, always starting from our group $\H=\langle h_1, h_2\rangle\leq \GL{s}(\Fqk)$, for each $i\in \{1,\dots, t\}$ and $j\in \{t+1,\dots, s\}$, we construct two sets of lines $\mathcal{A}_i$ and $\mathcal{B}_j$ of $\cG_{q^k}(1,s)$, such that $|\mathcal{A}_i|=|\mathcal{B}_j|=r$ and $\cG_{q^k}(1,s)=\mathcal{C}_i\cup \mathcal{A}_i\cup \mathcal{B}_j$. Afterward, we will use the field reduction technique to obtain a $k$-spread of $\Fq^n$.

Recall that the matrices $h_{1}$ and $h_{2}$ of $\H$ have multiplicative order $q^{kt}-1=r(q^{k}-1)$, with $\gcd(r,q^k-1)=1$. 
We are going to use the following subgroups of $\H$
\begin{equation}\label{eq:subgrupsH}
  \H_{2}
  :=
  \langle h_{2}^{q^{k}-1}\rangle, 
  \quad 
  \N
  :=
  \langle (h_1h_2)^r\rangle
  \quad \text{and} \quad 
  \T
  :=
  \langle h_{1}, h_{2}^{q^{k}-1}\rangle
  =
  \langle h_{1}\rangle\H_{2}.
\end{equation}

\begin{lemma}\label{lem:HNT}
Consider the groups $\H$, $\H_2$, $\N$ and $\T$ given in expression~\eqref{eq:subgrupsH}. 
One has: 
\begin{enumerate}
\item \label{HNT.1}
$|\H_{2}|=r$.
\item \label{HNT.2}
$\N\cap \T=\{I_s\}$ and $\H=\N \T$.
\end{enumerate}
\end{lemma}
\begin{proof}
\ref{HNT.1}
Immediate, because $o(h_2)=q^{kt}-1=(q^k-1)r$, with $\gcd(r,q^k-1)=1$.

\ref{HNT.2}
Notice that, from expression~\eqref{eq:ab}, the general matrix expression of elements of $\T$ will be 
\begin{equation}\label{eq:abT}
  h_1^a h_2^{(q^k-1)l}
  =
  \begin{pmatrix}
    C^a           & D_{a,(q^k-1)l} \\
    0_{t\times t} & \alpha^aC^{(q^k-1)l}
  \end{pmatrix}, 
\end{equation}
where $a\in\{1,\dots,q^{kt}-1\}$ and $l\in\{1,\dots, r\}$. Now, assume that $a$ and $l$ are such that  $h_1^{a}h_2^{(q^k-1)l} \in \N \cap \T$, . 
Since $\N=\langle \alpha I_s\rangle$ by Theorem~\ref{lem:stabH}, then $D_{a,(q^k-1)l}=0_{t\times t}$ and therefore, by Lemma~\ref{lem:ab}, we have that $a=(q^k-1)l$. 
Moreover, we have that 
\[
  C^a \in  \langle \alpha I_t\rangle \cap  \langle C \rangle = \{I_t\}
\]
since that $o(\langle \alpha I_s\rangle)=q^k-1$ and $o(\langle C \rangle)=r$ with $\gcd(q^k-1,r)=1$. 
Therefore, $C^a=I_t$ and then $o(C)$ must divide $a=(q^k-1)l$.
Consequently, $r$ must divide $l$. 
Thus, $r=l$ since $l\in\{1,\dots, r\}$ and  $a=(q^k-1)r=q^{kt}-1$. 
In this way we obtain that $h_1^{a}h_2^{(q^k-1)l}=(h_1h_2)^{(q^k-1)r}=I_s$ and $\N\cap \T=\{I_s\}$. 
As a consequence, 
\[
  |\N\T|=|\N|\cdot|\T|=(q^k-1)(q^{kt}-1)r=(q^{kt}-1)^2=|\H|
\]
and we conclude that $\H=\N\T$.
\end{proof}

Our interest for the subgroup $\T$ given in expression~\eqref{eq:subgrupsH} is explained in the following result. 

\bigskip

\begin{theorem}\label{th:CiT}
Consider the orbit code $\mathcal{C}_i$ defined in expression~\eqref{eq:Ci}, for $i\in \{1,\dots, t\}$ and the group $\T$ given in expression~\eqref{eq:subgrupsH}.
It follows that 
\[
  \mathcal{C}_i
  =
  \orb_{\T}(\rsp(\vec{e}_{i})).
\]
\end{theorem}
\begin{proof}
We have that
\[
  \mathcal{C}_i
  =
  \orb_{\H}(\rsp(\vec{e}_{i}))
  =
  \{\rsp(\vec{e}_{i})\cdot h \ | \ h \in \H\}.
\]
By Lemma~\ref{lem:HNT}, we have that $\H=\N\T$; therefore, given $h\in \H$, there exist $h_{\N} \in \N$ and $h_{\T} \in \T$ such that $h=h_{\N}h_{\T}$. 
Thus, we can express
\[
  \mathcal{C}_i
  =
  \left\{
    \left( \rsp(\vec{e}_{i})\cdot h_{\N} \right) \cdot h_{\T} 
    \ | \ 
    h_{\N} \in \N, \ h_{\T} \in \T 
  \right\}.
\]
Now, by Theorem \ref{lem:stabH}, $\N=\stab_{\H}(\rsp(\vec{e}_{i}))$ and we obtain that
\[
  \mathcal{C}_i
  =
  \left\{
    \rsp(\vec{e}_{i})\cdot h_{\T} 
    \ | \ 
    h_{\T} \in \T
  \right\}
  =
  \orb_{\T}(\rsp(\vec{e}_{i})).
\]
\end{proof}

Due to the fact that we can write $q^{kt}-1=(q^k-1)r$, we can form a partition of the set $\{1,\dots,q^{kt}-1\}$ through the sets $A_m=\{a_mr+m\ |\ 0\leq a_m\leq q^k-2\}$, for $m\in\{1,\dots,r\}$, obtaining that
\begin{equation}\label{eq:particio}
  \left\{
    1,\dots,q^{kt}-1
  \right\}
  =
  A_1\cup A_2\cup\dots\cup A_r.
\end{equation}

Now, according to Theorem~\ref{th:CiT}, for each $i\in\{1,\dots,t\}$, the orbit code $\mathcal{C}_i$ can be described as the set of lines of $\Fqk^s$ generated by the $i$-th row of each matrix of $\T$. 
We are going to use the partition of $\{1,\dots,q^{kt}-1\}$ provided in expression~\eqref{eq:particio} in order to analyse these lines. 

\bigskip

\begin{lemma} \label{lem:Bm}
For any integers $a\in\{1,\dots, q^{kt}-1\}$ and $l\in\{1,\dots,r\}$, consider the matrices $C^a$ and $D_{a,(q^k-1)l}$ given in expression~\eqref{eq:abT}. 
For any $m\in\{1,\dots,r\}$ it follows that
\begin{enumerate}
\item \label{lem:Bm.1}
$C^a=C^m$ if and only if $a\in A_m$.
\item \label{lem:Bm.2}
There exists $B_m\in \Fqk[M_t]$ such that $B_m\neq D_{a,(q^k-1)l}$, for all $a\in A_m$ and for all $l\in\{1,\dots,r\}$.
\end{enumerate}
\end{lemma}
\begin{proof}
\ref{lem:Bm.1}
If $C^a=C^m$ then $o(C)=r$ must divide $a-m$ and then $a\in A_m$. 
Now, if $a \in A_m$, then  there exists $a_m \in \{0, \ldots, q^k-2\}$  such that $a=a_mr+m$. 
Therefore, as $o(C)=r$, we have that $C^a=C^{a_mr+m}=C^m$.  

\ref{lem:Bm.2}
Notice that $\card{A_m}r=(q^k-1)r=q^{kt}-1$. 
Therefore, there exist at most $q^{kt}-1$ different matrices $D_{a,(q^k-1)l}$, when we consider all $a \in A_m$ and all $l\in\{1,\dots r\}$. 
Since $\Fqk[M_t]$ has $q^{kt}$ elements, we can find at least one matrix $B_m\in \Fqk[M_t]$ such that $B_m\neq D_{a,(q^k-1)l}$, for all $a\in A_m$ and for all $l\in\{1,\dots,r\}$.
\end{proof}

Next, for each $i\in\{1,\dots,t\}$, we use the matrices $B_m$ obtained in Lemma \ref{lem:Bm}, in order to define the following sets of lines of $\Fqk^s$ 
\begin{equation} \label{eq:Ai}
   \mathcal{A}_i:=\{\rsp((C^m | B_m)_i)\ |\ 1\leq m\leq r\}.
\end{equation}
Finally, we use the action of the group $\H_2$ on the lines of $\Fqk^s$ generated by the canonical vectors $\vec{e}_{j}$, for $j\in\{t+1,\dots,s\}$ and we consider the orbit codes
\begin{equation}\label{eq:codismenuts}
  \mathcal{B}_{j}
  :=
  \orb_{\H_{2}}(\rsp(\vec{e}_{j}))
  \subseteq
  \cG_{q^k}(1,s).
\end{equation}

Now, we have all the ingredients to complete each orbit code $\mathcal{C}_i$ until we get the whole space of lines $\cG_{q^k}(1,s)$ of $\Fqk^s$. 
To this end, we prove now the following technical lemma.

\begin{lemma} \label{lema:BCD}
Consider $a \in \{1, \ldots, q^{kt}-1\}$, $m,l \in \{1, \ldots,r\}$ and $i \in \{1, \ldots,t\}$. 
Then, the matrices $\left( C^m \ | \ B_m\right)$ and $\left( C^a \ | \ D_{a,(q^k-1)l}\right)$ are different if and only if the $i$-th row of the matrices $\left(C^m \ | \ B_m\right)$ and $ \left( C^a \ | \ D_{a,(q^k-1)l}\right)$ are different. 
\end{lemma}
\begin{proof}
Suppose that $\left( C^m \ | \ B_m\right)$ and $\left( C^a \ | \ D_{a,(q^k-1)l}\right)$ are different matrices and, arguing by contradiction, that their $i$-th rows are equal, that is, $\left( C^m \ | \ B_m\right)_i = \left( C^a \ | \ D_{a,(q^k-1)l}\right)_i$. In particular, this means that the $i$-th row of the matrix $C^m-C^a$ is null. Due to $C^m-C^a \in \Fqk[M_t]$, we have that $C^m=C^a$. Then $a \in A_m$, by Lemma~\ref{lem:Bm}. Moreover, since the $i$-th row of the matrix $B_m-D_{a,(q^k-1)l}\in \Fqk[M_t] $ is null, it follows that $B_m=D_{a,(q^k-1)l}$, but this is a contradiction with the choice of $B_m$. 

The other implication of the statement is trivial.    
\end{proof}

\begin{theorem} \label{th:B}
Consider integers $i\in\{1,\dots,t\}$, $j\in\{t+1,\dots,s\}$ and the one dimensional codes of $\Fqk^s$, $\mathcal{C}_i$, $\mathcal{A}_i$  and $\mathcal{B}_{j}$ given by expressions \eqref{eq:Ci}, \eqref{eq:Ai}, and \eqref{eq:codismenuts}, respectively. 
One has:
\begin{enumerate}
\item \label{th:B.1}
$|\mathcal{A}_i|=|\mathcal{B}_{j}|=r$.
\item \label{th:B.2}
$\cG_{q^k}(1,s)=\mathcal{C}_i\cup \mathcal{A}_i \cup \mathcal{B}_{j}$ with empty pairwise intersection.
\end{enumerate}
\end{theorem}
\begin{proof}
\ref{th:B.1}
Notice that $|\mathcal{A}_i|=r$, for any $i\in\{1,\dots,t\}$, by Lemma \ref{lema:BCD}. 
Now, consider $j \in\{t+1,\dots,s\}$ and let us see that $\stab_{\H_{2}}(\rsp(\vec{e}_{j}))=\{I_s\}$. 
Let $h_2^{(q^k-1)l}\in \stab_{\H_{2}}(\rsp(\vec{e}_{j}))$, for some $l\in\{1,\dots, r\}$. 
Taking into account the matrix expression of elements of $\H_2$, we can write 
\begin{equation} \label{matrixh1}
  h_2^{(q^k-1)l}
  =
  \begin{pmatrix}
    I_{t}         & D_{0,(q^k-1)l} \\
    0_{t\times t} &  C^{(q^k-1)l}
  \end{pmatrix}.
\end{equation} 
Now, since $j\in\{t+1,\dots, s\}$, we can write $j=t+j_1$, for some $j_1\in\{1,\dots,t\}$ and consider the $j_1$-th canonical vector $\vec{e}_{j_1}$ of $\Fqk^t$. 
In this way, using Remark \ref{rem:singer}, we obtain that  
\[
  C^{(q^k-1)l}
  \in 
  \stab_{\langle C^{q^k-1}\rangle} (\rsp(\vec{e}_{j_1}))
  \subseteq 
  \stab_{\langle M_t\rangle} (\rsp(\vec{e}_{j_1}))
  =
  \langle \alpha I_t\rangle.
\]
But
\[
  \stab_{\langle C^{q^k-1}\rangle} (\rsp(\vec{e}_{j_1}))
  \subseteq 
  \langle C\rangle \cap \langle \alpha I_t\rangle=\{I_t\},
\]
since $\gcd(r,q^k-1)=1$. 
Thus $l=r$ and $\stab_{\H_{2}}(\rsp(\vec{e}_{j}))=\{I_s\}$. 
That is, $|\mathcal{B}_{j}|=|\H_{2}|=r$.

\ref{th:B.2}
Let $i\in\{1,\dots,t\}$ and $j\in\{t+1,\dots,s\}$. 
We prove that $\mathcal{C}_i$, $\mathcal{A}_i$ and $\mathcal{B}_{j}$ have an empty pairwise intersection.
\begin{itemize}
\item 
$\mathcal{C}_i \cap \mathcal{B}_{j} = \emptyset$. 
Let $u \in \mathcal{C}_i \cap \mathcal{B}_{j}$.     
Since $u \in \mathcal{C}_i$, then $u=\rsp \left( h_1^{a}h_2^b\right)_{i}$ for some $a,b \in \{1, \ldots,q^{kt}-1\}$;
that is, $u$ is the line of $\Fqk^s$ generated by the $i$-th row of the matrix given in expression~\eqref{eq:ab}.
However, since $u \in \mathcal{B}_{j}$, then
$u$ is also the line of $\Fqk^s$ generated by the $j$-th row of a matrix like the one given in expression~\eqref{matrixh1}.
Thus, since $i\in \{1,\dots, t\}$ and $j\in\{t+1,\dots, s\}$, we deduce that the $i$-th row of $\alpha^bC^a$ must be zero, which is not possible, since this is a non-zero matrix in the field $\Fqk[M]$. Therefore, $\mathcal{C}_i \cap \mathcal{B}_{j}= \emptyset$.
\item $\mathcal{A}_i \cap \mathcal{B}_{j} = \emptyset$. 
The reasoning is analogous to the previous case, because the first block matrix of size $t \times t$ of matrices associated to $\mathcal{A}_i = \{\rsp((C^m \ | \ B_m)_i) \ | \ 1\leq m\leq r\}$ is always a power of $C$.
\item $\mathcal{C}_i \cap \mathcal{A}_{i} = \emptyset$.
Suppose that $u \in \mathcal{C}_i \cap \mathcal{A}_{i}$ with $i \in \{1, \ldots, t\}$.
On the one hand, since $\mathcal{C}_i = \orb_{\T}(\rsp(\vec{e}_{i}))$ by Theorem~\ref{th:CiT}, we know that $u$ is the line of $\Fqk^s$  generated by the $i$-th row of the matrix 
\begin{equation*}
  h_1^a h_2^{(q^k-1)l}
  =
  \begin{pmatrix}
    C^a           & D_{a,(q^k-1)l} \\
    0_{t\times t} & \alpha^aC^{(q^k-1)l}
  \end{pmatrix},
\end{equation*}
for some integers $a\in\{1,\dots,q^{kt}-1\}$ and $l\in\{1,\dots, r\}$.

On the other hand, as $u \in \mathcal{A}_{i}$, we have that the $i$-th row of $h_1^a h_2^{(q^k-1)l}$ must be proportional to the $i$-th row of $(C^m \ | \ B_m)$ for some $m \in \{1, \ldots, r\}$. 
In other words, 
\[
  (C^a \ | \ D_{a,(q^k-1)l})_i
  =
  (\alpha^d(C^m \ | \ B_m))_i,
\]
for some $\alpha^d \in \Fqk^{\star}$. 
In particular, this means that the $i$-th row of $C^a - \alpha^dC^m$ must be zero. 
Since this matrix is in the field $\Fqk[M_t]$,  we conclude that $C^a=\alpha^dC^m$. 
Therefore, $C^{a-m} = \alpha^d I_t \in \langle \alpha I_t\rangle \cap \langle C \rangle =\{I_t\}$ and then 
$\alpha^d=1$ and $a=m$. 
We have that
\[
  (C^a \ | \ D_{a,(q^k-1)l})_i
  =
  (C^m \ | \ D_{m,(q^k-1)l})_i
  =
  (C^m \ | \ B_m)_i.
\]

From Lemma~\ref{lema:BCD}, we obtain that the whole matrices must be equal, that is, $(C^m \ | \ B_m) = (C^m \ | \ D_{m,(q^k-1)l})$. 
But this is not possible because $m\in A_m$ and  $B_m \neq D_{a,(q^k-1)l}$ for all $a \in A_m$ and for all $l\in \{1, \ldots, r\}$, by Lemma \ref{lem:Bm}.
\end{itemize}
\end{proof}

Next, just as we construct the  $k$-partial spreads $\bar{\cC_i}$ of $\Fq^n$ in expression \eqref{eq:Cibarra}, we construct now $k$-partial spreads from the codes $\mathcal{A}_{i}$ and $\mathcal{B}_{j}$, for any $i\in\{1,\dots,t\}$ and $j\in\{t+1,\dots,s\}$. Denote $\bar{\H}_{2}:=\psi(\H_{2})$. 
We have that 
\begin{align}
  \bar{\mathcal{A}}_{i}
  & :=
  \varphi(\mathcal{A}_i)
  =
  \{\varphi(\rsp((C^m \ | \ B_m)_i))\ |\ 1\leq m\leq r\}, \label{eq:Abarra} \\
  \bar{\mathcal{B}}_{j}
  & :=
  \varphi(\mathcal{B}_{j})
  =
  \varphi(\orb_{\H_{2}}(\rsp(\vec{e}_{j}))) \nonumber  \\
  & =
  \orb_{\psi(\H_{2})}(\varphi(\rsp(\vec{e}_{j})))
  =
  \orb_{\bar{\H}_{2}}(\mathcal{U}_{k,j}) \label{eq:Bbarra}
\end{align}
are $k$-partial spreads of $\Fq^n$. 

Finally, we present our last result, which describes how we can obtain the Desarguesian spread given in expression \eqref{def:spreadS} from these three partial spreads. 

\bigskip

\begin{theorem}
For any $i\in\{1,\dots,t\}$ and $j\in\{t+1,\dots,s\}$, consider the orbit codes $\bar{\cC}_{i}$, $\bar{\mathcal{A}}_{i}$, and $\bar{\mathcal{B}}_{j}$ defined in expressions \eqref{eq:Cibarra}, \eqref{eq:Abarra}, and \eqref{eq:Bbarra}, respectively.
Then, the code $\bar{\cC_{i}}\cup \bar{\mathcal{A}}_{i}\cup \bar{\mathcal{B}}_{j}$ is a $k$-spread of $\Fq^n$.
\end{theorem}

\section{Conclusions}
\label{sec:concl}

In this paper, we have dealt with the orbital construction of a $k$-dimensional spread in $\Fq^n$, where $n$ is an even number and $k$ divides $n$, using a non-cyclic Abelian group. 
Our results generalise those obtained by Chen and Liang
in \cite{Chen2021b} for $\Fq^{2k}$. 
However, the techniques we have used are new and not easily detached from this work. 
We have proceeded in two stages. 
First, we have obtained a family of orbit codes of $\Fq^n$ having dimension $k$ and maximum distance. 
For this, we have constructed a non-cyclic Abelian group of $\GL{n}(\Fq)$ and selected appropriate $k$-dimensional subspaces of $\Fq^n$ using the field reduction technique. 
Afterwards, from each previously given orbit code, we have managed to achieve a $k$-spread of $\Fq^n$ by adding two specific subsets of subspaces, one of them with an orbital structure. 

\bmhead{Acknowledgements}

This work was partially supported by the Spanish I+D+i project PID2022-142159OB-I00 of the Ministerio de Ciencia e Innovaci\'{o}n, I+D+i project
CIAICO/2022/167 of the Generalitat Valenciana, and the I+D+i project VIGROB-287 of the Universitat d`Alacant.



\end{document}